\documentclass[hidelinks,onefignum,onetabnum]{siamart251104}

\usepackage{amsfonts}
\usepackage{amsmath,amssymb}
\usepackage{algorithm}
\usepackage{algpseudocode}
\crefname{algorithm}{Algorithm}{Algorithms}
\Crefname{algorithm}{Algorithm}{Algorithms}
\hypersetup{pdftitle={Stable complete coordinates for multisets of points via basic r-symmetric tropical polynomials},
  pdfauthor={S. Kubo}}

\newcommand{\rmath}{\ensuremath{r}}
\newcommand{\R}{\mathbb{R}}
\newcommand{\N}{\mathbb{N}}
\newcommand{\supp}{\operatorname{supp}}
\newcommand{\proj}{\operatorname{proj}}

\newsiamremark{remark}{Remark}
\newsiamthm{example}{Example}
\theoremstyle{nonumberplain}\newtheorem{theoremstar}{Theorem}\theoremstyle{plain}

\headers{Tropical coordinates for multisets of points}{S. Kubo}

\title{Stable complete coordinates for multisets of points\\
 via basic \rmath-symmetric tropical polynomials\thanks{Submitted to the
 editors \today.\funding{This work
 was partially supported by MEXT Leading Initiative for Excellent Young
 Researchers Grant Number JPMXS0320200347 and JSPS KAKENHI Grant Number
 JP26K17039.}}}

\author{Susumu Kubo\thanks{Faculty of Informatics, Showa Women's
 University, Tokyo, Japan (\email{s-kubo@swu.ac.jp}),
 \url{https://orcid.org/0000-0003-1780-9677}.}}

\begin{document}
\maketitle

\begin{abstract}
A multiset of $n$ unordered points in $\R^r$ --- a point cloud, or, for
$r=2$, a persistence barcode of birth--death pairs --- is a point of
the orbit space $\R^{nr}/S_n$ for the symmetric group $S_n$ permuting
the rows of an $n\times r$ matrix; a separating family of invariants on
this space is exactly a complete set of permutation-independent
coordinates. We provide one that is explicit, small, and stable, in the
max-plus (tropical) setting: for all $n\ge1$ and $r\ge1$, the
$\binom{n+r}{r}$ \emph{basic $r$-symmetric tropical polynomials}, of
degree at most $n$, separate the orbits of $S_n$ on $\R^{nr}$. This
settles in full a problem left open in [Kubo, J.~Pure Appl.\ Algebra
223 (2019) 72--85], where separation was known only for $r=2$ and
special cases of $r\ge3$, and yields a family far smaller and of lower
degree than the general separating sets from Derksen's recent theory of
tropical invariants for permutation actions ($nr+(nr)!/n!$ invariants
of degree $O(n^2r^2)$). The proof is elementary and constructive: the
basic values are identified with a transportation problem, and the
multiset is recovered from the dual by an explicit algorithm. We
further show the coordinate map is a bi-Lipschitz embedding for all
$n$ and $r$, being an injective max filter bank (via the bi-Lipschitz
theory of max filtering), with an explicit Lipschitz constant for the
forward bound and a fully explicit, dimension-free distortion when
$r=1$. Finally we determine when the pairwise values suffice (exactly
$n\le3$) and show that invariants on at least three columns and of
degree less than $n$ are necessary in general, the obstruction being a
standard non-uniqueness configuration from discrete tomography.
\end{abstract}



\section{Introduction}\label{sec:intro}

Many objects in applied geometry are \emph{multisets} of points: a
point cloud is a multiset of $n$ points in $\R^r$, and the special case
$r=2$ includes the persistence diagram (barcode) of a filtration, a
multiset of birth--death pairs in $\R^2$. Recording such
an object as a matrix of coordinates and forgetting the order of the
points amounts to passing to the orbit space $\R^{nr}/S_n$ of the
symmetric group acting by row permutations. A \emph{separating} family
of invariants on this space is then precisely a complete set of
permutation-independent coordinates, turning an unordered configuration
into a finite vector of numbers; if the coordinates are moreover stable
under perturbations, they are well suited to downstream computation and
learning. This is the viewpoint of coordinates on barcode spaces
\cite{CK}, of separating invariants for group actions \cite{Derksen},
and of stable group-invariant embeddings such as max filtering
\cite{CIMP}. We provide such coordinates --- small, explicit, complete,
and Lipschitz-stable --- in the max-plus (tropical) setting.

Tropical algebra, or max-plus algebra, replaces the sum $a+b$ of two
real numbers by the maximum $\max(a,b)$ and the product $ab$ by the
sum $a+b$. Invariant
theory in the tropical setting is a young subject. Matsui
\cite{Matsui} ultradiscretized Weyl's basic invariants, recovered the
order statistics in the single-attribute case, and anticipated that in
general the resulting tropical invariants determine a point multiset up
to permutation; Carlsson and Kali\v{s}nik \cite{CK} introduced
symmetric and $r$-symmetric tropical polynomials and rational
functions, and a separating family among them, motivated by coordinates
on spaces of persistence barcodes. This \emph{basic} family for the row
permutation action was placed on a formal footing in \cite{Kubo}, whose
setting we now recall. Separating these orbits is a problem of interest
well beyond tropical algebra, from persistence barcodes \cite{CK} to
the general theory of separating invariants for group actions
\cite{Derksen}.

Arrange $nr$ variables in an $n\times r$ matrix
\[
M=\begin{pmatrix}
m_{11} & \cdots & m_{1r}\\
\vdots & & \vdots\\
m_{n1} & \cdots & m_{nr}
\end{pmatrix},
\]
and let the symmetric group $S_n$ act by permuting the rows. A
tropical polynomial in these variables is \emph{$r$-symmetric} if it is
invariant under this action. The \emph{basic $r$-symmetric tropical
polynomials} $b_{q_1,\dots,q_r}$, indexed by nonnegative integers
$q_1,\dots,q_r$ whose sum is at most $n$, are the ultradiscretizations
(tropicalizations) of Weyl's
typical basic invariants \cite{Weyl}; following Matsui \cite{Matsui},
they were defined in \cite[Definition~3]{Kubo}; their
values admit the combinatorial description \eqref{eq:bvalue} given in
\cref{sec:prelim}.
The main theorem of \cite{Kubo} states that for $r=2$ the basic
polynomials separate the orbits of $S_n$ on $\R^{2n}$, i.e.\ they give
coordinates on $\R^{2n}/S_n$, improving the separating family of
\cite{CK}; for $r\ge3$, separation was established whenever some
column, or some difference of two columns, has pairwise distinct
entries (\cite[Propositions~13, 14]{Kubo}), and the general
case was raised as an open problem in the concluding discussion.

Interest in such questions has been renewed by recent work of Derksen
\cite{Derksen} on tropical invariants of arbitrary permutation group
actions. Two of its results are directly relevant. First,
for a subgroup $G\le S_N$ acting on $N$ variables, the semiring of
invariant tropical polynomials is finitely generated if and only if $G$
is generated by transpositions; for the row permutation action of
$S_n$ on $N=nr$ variables with $r\ge2$ and $n\ge2$ this fails, so
\emph{no finite generating set exists}, a fact already noted for
$r$-symmetric tropical polynomials by Carlsson and Kali\v{s}nik
\cite{CK}: a finite family of invariants
cannot generate all invariants, but it can still \emph{separate} the
orbits, which is what we establish for the basic family. Second, for every
permutation action there exists a finite separating family, of
cardinality $N+N!/|G|$ and degree at most $\max(N,\binom N2)$; for
the row permutation action this gives $nr+(nr)!/n!$ invariants of
degree $O(n^2r^2)$. Against this benchmark, the basic family is
remarkably small: it consists of $\binom{n+r}{r}$ invariants of degree
at most $n$. Derksen's separating family is moreover a bi-Lipschitz
embedding \cite{Derksen}; the basic family shares this stability
(\cref{cor:bilipschitz}) while being far smaller and of lower degree.

In this paper we resolve the open problem of \cite{Kubo} in full
generality.

\begin{theoremstar}[\cref{thm:main}]
For all $n\ge1$ and $r\ge1$, the basic $r$-symmetric tropical
polynomials separate the orbits of $S_n$ on $\R^{nr}$.
\end{theoremstar}

The resulting coordinate map, which sends an $n\times r$ matrix $M$ to
the vector of its basic values
$(b_c(M))_{c\in\N^r,\,|c|\le n}$, is moreover \emph{stable}: it is
Lipschitz with respect to the optimal-matching
metric on $\R^{nr}/S_n$ (\cref{prop:stability}), and in fact a
bi-Lipschitz embedding for all $n$ and $r$
(\cref{cor:bilipschitz}), the latter via the bi-Lipschitz theory of
max filter banks \cite{BT,CIMP}, with an explicit dimension-free
constant when $r=1$. Thus the $\binom{n+r}{r}$ basic
values form not merely a complete set of coordinates but a stable one,
as is desirable for coordinates used as features of point clouds and
barcodes. Separation here is, furthermore, effective: the orbit is not merely
determined by the basic values but can be computed from them, by the
explicit reconstruction algorithm of \cref{alg:reconstruct}.

The proof occupies \cref{sec:main} and is elementary. The basic
value $b_c(M)$ is the optimal value of a transportation problem that
assigns rows of $M$ to columns without repetition. Guided by linear
programming duality (though no duality theory is invoked), we introduce
the transform
\[
G_M(\lambda)=\sum_{i=1}^{n}\max\bigl(0,\;
\max_{1\le\alpha\le r}(m_{i\alpha}-\lambda_\alpha)\bigr),
\qquad\lambda\in\R^r,
\]
where $m_{i\alpha}$ is the $(i,\alpha)$-entry of $M$, and show that the basic
values determine $G_M$ (\cref{lem:duality}), that $G_M$ determines
the counting function $\lambda\mapsto\#\{i:m_i\le\lambda\}$ of the rows
$m_1,\dots,m_n\in\R^r$ of $M$
(\cref{lem:counting}), and that the counting function determines
the multiset of rows (\cref{lem:cdf}). The argument is
constructive and yields an explicit procedure that reconstructs the
orbit from the $\binom{n+r}{r}$ basic values
(\cref{rem:algorithm}).

\Cref{sec:pairwise} analyses how much of the basic family is
really needed. The values supported on two columns
$\{\alpha,\beta\}$ carry exactly the orbit of the column pair
$(x^{(\alpha)},x^{(\beta)})$, by the case $r=2$ of our main theorem
(which itself recovers the $r=2$ result of \cite{Kubo}). We
prove that these pairwise values separate all orbits if and only if
$n\le3$ (\cref{prop:n3}); the obstruction at $n=4$ is a standard
non-uniqueness configuration of discrete tomography: two distinct sets
with the same X-rays in the three coordinate directions \cite{GG}. We further
show that the top-degree basic values encode the orbit of the
difference matrix together with a single column sum
(\cref{prop:topdegree}), and we exhibit an
eight-row configuration for which the pairwise and top-degree values
all agree while the orbits differ --- so invariants supported on at
least three columns and of degree less than $n$ are necessary
(\cref{ex:switching}). \Cref{sec:discussion} concludes
with remarks and open questions.

\section{Preliminaries}\label{sec:prelim}

Throughout, $\N=\{0,1,2,\dots\}$ denotes the nonnegative integers,
$e_1,\dots,e_r$ the standard basis of $\R^r$, and the matrix dimensions
$n$ and $r$ are positive integers.

In the tropical semiring $(\R\cup\{-\infty\},\oplus,\odot)$, known as
max-plus algebra (see e.g.\ \cite{Butkovic,IMS}), addition and
multiplication are defined by $a\oplus b=\max(a,b)$ and
$a\odot b=a+b$, with additive identity $-\infty$ and
multiplicative identity $0$. \emph{Tropical polynomials} in $N$ variables are the
functions $\R^N\to\R$ obtained from tropical monomial expressions
$a\odot x_1^{i_1}\odot\cdots\odot x_N^{i_N}$ ($a\in\R$,
$i_1,\dots,i_N\in\N$) by finitely many applications of $\oplus$; two expressions
defining the same function are identified \cite{CK}. The
\emph{degree} of a tropical polynomial is the largest total degree
$i_1+\dots+i_N$ of a monomial needed to represent it.

Throughout, $M$ denotes an $n\times r$ matrix of
variables (or of
reals) with columns $x^{(1)},\dots,x^{(r)}\in\R^n$ and rows
$m_1,\dots,m_n\in\R^r$, with entries $m_{i\alpha}$. The group
$S_n$ acts on $\R^{nr}$ by permuting the rows; we write
$[M]\in\R^{nr}/S_n$ for the orbit of $M$ and freely identify $[M]$
with the \emph{multiset} $\{m_1,\dots,m_n\}$ of rows, viewed as a
multiset of $n$ points of $\R^r$. A tropical polynomial in the entries
of $M$ is \emph{$r$-symmetric} if it is constant on the orbits.

The basic $r$-symmetric tropical polynomials were defined in
\cite[Definition~3]{Kubo}; for our purposes only
their values are needed. For $c=(c_1,\dots,c_r)\in\N^r$ with
$|c|:=c_1+\dots+c_r\le n$, the basic polynomial
$b_c=b_{c_1,\dots,c_r}$ has the value
\begin{equation}\label{eq:bvalue}
b_c(M)=\max\Bigl\{\,\sum_{\alpha=1}^{r}\sum_{i\in S_\alpha}
m_{i\alpha}\;:\;S_1,\dots,S_r\subseteq\{1,\dots,n\}
\text{ pairwise disjoint},\ |S_\alpha|=c_\alpha\Bigr\},
\end{equation}
and $b_0(M)=0$. Thus $b_c$ is an $r$-symmetric tropical
polynomial of degree $|c|$, and the family
$\{b_c:c\in\N^r,\ |c|\le n\}$ has $\binom{n+r}{r}$ members. We call the
numbers $b_c(M)$, for $c\in\N^r$ with $|c|\le n$, the \emph{basic
values} of $M$; being $r$-symmetric, they depend only on the orbit
$[M]$. We write
$\supp(c)=\{\alpha:c_\alpha\neq0\}$ for the support of $c$. For $r=1$
the basic polynomials are the elementary symmetric tropical
polynomials $b_k(x)=\max_{|S|=k}\sum_{i\in S}x_i$ \cite{CK,Matsui}.

We say that a family $\mathcal F$ of $r$-symmetric tropical
polynomials \emph{separates orbits} if for all
$M,\bar M\in\R^{nr}$, the equalities $f(M)=f(\bar M)$ for all
$f\in\mathcal F$ imply $[M]=[\bar M]$.

\section{The main theorem}\label{sec:main}

We begin by describing the optimization problem underlying the
value \eqref{eq:bvalue}. Fix $c\in\N^r$ with $|c|\le n$, and regard the
$n$ rows of $M$ as agents and the $r$ columns as tasks: assigning row
$i$ to column $\alpha$ yields the profit $m_{i\alpha}$. We look for an
assignment in which each row is used at most once and column $\alpha$
receives exactly $c_\alpha$ rows, so as to maximize the total profit.
Writing $S_\alpha\subseteq\{1,\dots,n\}$ for the set of rows assigned
to column $\alpha$, this is the assignment-type transportation problem
\[
\text{maximize}\quad
\sum_{\alpha=1}^{r}\sum_{i\in S_\alpha}m_{i\alpha}
\qquad\text{subject to}\qquad
\begin{aligned}
&|S_\alpha|=c_\alpha\quad(\alpha\in\{1,\dots,r\}),\\
&S_1,\dots,S_r\ \text{pairwise disjoint},
\end{aligned}
\]
and its optimal value is, by definition \eqref{eq:bvalue}, exactly the
basic value $b_c(M)$. Guided by linear programming duality for this
problem --- though the proofs below are elementary and invoke no
duality theory --- we attach to a price vector $\lambda\in\R^r$ the
quantities
\[
u_i(\lambda)\;=\;\max\bigl(0,\ \max_{1\le\alpha\le r}
(m_{i\alpha}-\lambda_\alpha)\bigr)\qquad(i\in\{1,\dots,n\}),
\qquad
G_M(\lambda)\;=\;\sum_{i=1}^{n}u_i(\lambda).
\]
Here $\lambda_\alpha$ acts as a price charged for occupying column
$\alpha$: the term $u_i(\lambda)$ is the best net profit row $i$ can earn from a
column, or $0$ if every column would give negative net profit, in which
case row $i$ is left unassigned; and $G_M(\lambda)$ sums these over all
rows.\footnote{%
Explicitly, encode an assignment by variables $z_{i\alpha}\in\{0,1\}$
($z_{i\alpha}=1$ when row $i$ is assigned to column $\alpha$); then
\eqref{eq:bvalue} is the linear program $\max\sum_{i,\alpha}
m_{i\alpha}z_{i\alpha}$ subject to $\sum_\alpha z_{i\alpha}\le1$ (one
constraint per row $i$), $\sum_i z_{i\alpha}=c_\alpha$ (one per column
$\alpha$), and $z_{i\alpha}\ge0$; its constraint matrix is totally
unimodular, so the optimum is integral and equals $b_c(M)$. The dual
assigns a multiplier $u_i\ge0$ to each row constraint and
$\lambda_\alpha\in\R$ to each column constraint, giving
$\min\bigl\{\sum_i u_i+\langle c,\lambda\rangle:\
u_i+\lambda_\alpha\ge m_{i\alpha}\ (\forall i,\alpha),\ u_i\ge0\bigr\}$.
Minimizing over $u$ for fixed $\lambda$ forces $u_i=u_i(\lambda)$ and
reduces the dual objective to $G_M(\lambda)+\langle c,\lambda\rangle$;
strong duality then reads
$b_c(M)=\min_\lambda\bigl(G_M(\lambda)+\langle c,\lambda\rangle\bigr)$,
of which \cref{lem:duality} is the elementary, duality-free
reformulation.}

\begin{lemma}\label{lem:duality}
For every $\lambda\in\R^r$,
\[
G_M(\lambda)\;=\;\max\bigl\{\,b_c(M)-\langle c,\lambda\rangle\;:\;
c\in\N^r,\ |c|\le n\,\bigr\}.
\]
In particular, the values of the basic
polynomials at $M$ determine the function $G_M:\R^r\to\R$.
\end{lemma}

\begin{proof}
($\ge$) Let $c\in\N^r$ with $|c|\le n$, and let $S_1,\dots,S_r$ be
pairwise disjoint sets with $|S_\alpha|=c_\alpha$. Then
\[
\sum_{\alpha=1}^{r}\sum_{i\in S_\alpha}m_{i\alpha}
-\langle c,\lambda\rangle
=\sum_{\alpha=1}^{r}\sum_{i\in S_\alpha}(m_{i\alpha}-\lambda_\alpha)
\le\sum_{\alpha=1}^{r}\sum_{i\in S_\alpha}u_i(\lambda)
\le\sum_{i=1}^{n}u_i(\lambda)=G_M(\lambda);
\]
the first inequality holds term by term by the definition of
$u_i(\lambda)$, and the second uses $u_i(\lambda)\ge0$ together with
the disjointness of $S_1,\dots,S_r$. Maximizing over
$S_1,\dots,S_r$ gives
$b_c(M)-\langle c,\lambda\rangle\le G_M(\lambda)$.

($\le$) For each $i$ with $u_i(\lambda)>0$, let $\alpha_i$ be the
smallest index $\alpha$ with
$m_{i\alpha}-\lambda_\alpha=u_i(\lambda)$ (such an $\alpha$ exists, as
$u_i(\lambda)>0$ equals $\max_\alpha(m_{i\alpha}-\lambda_\alpha)$), and set
\[
S_\alpha=\{\,i\;:\;u_i(\lambda)>0\ \text{and}\ \alpha_i=\alpha\,\}
\qquad(\alpha\in\{1,\dots,r\}).
\]
These sets are pairwise disjoint; put $c=(|S_1|,\dots,|S_r|)$, so
$c\in\N^r$ and $|c|\le n$. Then
\[
b_c(M)-\langle c,\lambda\rangle
\;\ge\;\sum_{\alpha=1}^{r}\sum_{i\in S_\alpha}
(m_{i\alpha}-\lambda_\alpha)
\;=\sum_{i\,:\,u_i(\lambda)>0}u_i(\lambda)\;=\;G_M(\lambda).
\]
The inequality holds because the $S_\alpha$ are disjoint with
$|S_\alpha|=c_\alpha$, so
$\sum_\alpha\sum_{i\in S_\alpha}m_{i\alpha}\le b_c(M)$; the first equality uses $m_{i\alpha}-\lambda_\alpha
=u_i(\lambda)$ for $i\in S_\alpha$, by the choice of $\alpha_i$,
together with the fact that the $S_\alpha$ partition
$\{i:u_i(\lambda)>0\}$; and the last equality holds since
$u_i(\lambda)=0$ for every other $i$.
\end{proof}

\begin{lemma}\label{lem:counting}
The function $G_M$ determines the counting function
\[
N_M(\lambda)\;=\;\#\{\,i\;:\;m_{i\alpha}\le\lambda_\alpha
\ \text{for all}\ \alpha\,\}\qquad(\lambda\in\R^r).
\]
Explicitly, with $\mathbf 1=(1,\dots,1)\in\R^r$,
\[
N_M(\lambda)\;=\;n-\lim_{t\downarrow0}
\frac{G_M(\lambda)-G_M(\lambda+t\mathbf 1)}{t}.
\]
\end{lemma}

\begin{proof}
Put $\mu_i=\max_{\alpha}(m_{i\alpha}-\lambda_\alpha)$, so that
$u_i(\lambda+t\mathbf 1)=\max(0,\mu_i-t)$ for all $t\in\R$, and
$\mu_i\le0$ if and only if $m_i\le\lambda$ componentwise. Let
$t_0=\min\{\mu_i:\mu_i>0\}$, with $t_0=\infty$ if no $\mu_i$ is
positive. For $0<t\le t_0$,
\[
G_M(\lambda)-G_M(\lambda+t\mathbf 1)
=\sum_{i=1}^{n}\bigl[\max(0,\mu_i)-\max(0,\mu_i-t)\bigr]
=t\cdot\#\{\,i:\mu_i>0\,\},
\]
because $t\le t_0$ means no $\mu_i$ lies in $(0,t)$, so each summand is
$0$ when $\mu_i\le0$ and $t$ when $\mu_i\ge t$, the latter occurring for
exactly the indices with $\mu_i>0$. Dividing by $t$ and letting
$t\downarrow0$ gives
$\#\{i:\mu_i>0\}=n-N_M(\lambda)$.
\end{proof}

\begin{lemma}\label{lem:cdf}
A multiset $\{m_1,\dots,m_n\}$ of points of $\R^r$ is determined by
its counting function $N_M$. Explicitly, for $v\in\R^r$ and
$S\subseteq\{1,\dots,r\}$ write
$\mathbf 1_S=\sum_{\alpha\in S}e_\alpha$; then for all sufficiently
small $\varepsilon>0$,
\[
\#\{\,i\;:\;m_i=v\,\}\;=\;
\sum_{S\subseteq\{1,\dots,r\}}(-1)^{|S|}\,
N_M(v-\varepsilon\mathbf 1_S).
\]
\end{lemma}

\begin{proof}
By inclusion--exclusion, the right-hand side counts the indices $i$
with $v_\alpha-\varepsilon<m_{i\alpha}\le v_\alpha$ for all $\alpha$.
Provided $\varepsilon$ is smaller than every nonzero value of
$|m_{i\alpha}-v_\alpha|$, these are exactly the indices with $m_i=v$.
\end{proof}

\begin{theorem}\label{thm:main}
Let $n\ge1$ and $r\ge1$. The basic $r$-symmetric tropical polynomials
$b_c$, ranging over $c\in\N^r$ with $|c|\le n$, separate the orbits of
$S_n$ on $\R^{nr}$. Consequently, the map
$M\mapsto(b_c(M))_{c\in\N^r,\,|c|\le n}$ induces an injection of
$\R^{nr}/S_n$ into $\R^{\binom{n+r}{r}}$.
\end{theorem}

\begin{proof}
For separation (injectivity of the map on the quotient), suppose that
$b_c(M)=b_c(\bar M)$ for all $c$. By
\cref{lem:duality}, $G_M=G_{\bar M}$ as functions on $\R^r$. By
\cref{lem:counting}, $N_M=N_{\bar M}$. By \cref{lem:cdf},
the row multisets of $M$ and $\bar M$ coincide, that is,
$[M]=[\bar M]$. The basic polynomials are $r$-symmetric, so they are
constant on each orbit; hence the map $M\mapsto(b_c(M))_c$ descends to
$\R^{nr}/S_n$, and the separation just proved is exactly its
injectivity, giving the asserted injection.
\end{proof}

\begin{remark}\label{rem:r1}
For $r=1$ the transform reduces to
$G_M(\lambda)=\sum_i\max(0,x_i-\lambda)$, which by
\cref{lem:counting} determines the counting function
$N_M(\lambda)=\#\{i:x_i\le\lambda\}$ ($n$ times the empirical
distribution function of $x_1,\dots,x_n$) and hence the multiset
$\{x_1,\dots,x_n\}$;
\cref{thm:main} then recovers the elementary fact that the partial
sums of the order statistics determine a real $n$-tuple up to
permutation. These partial sums are exactly the elementary symmetric
tropical polynomials \cite{CK,Matsui}.
\end{remark}

\begin{remark}[Effective reconstruction]\label{rem:algorithm}
The proof is constructive. Given the $\binom{n+r}{r}$ basic
values $b_c$ of an unknown $M$, one evaluates
$G_M(\lambda)=\max_c\,(b_c-\langle c,\lambda\rangle)$ at any desired
$\lambda$, obtains $N_M$ by \cref{lem:counting}, and reads off the
rows by \cref{lem:cdf}. The entries of column $\alpha$ are themselves
read off from the basic values: the $k$th largest entry of column
$\alpha$ is $y_k=b_{ke_\alpha}-b_{(k-1)e_\alpha}$ (with $b_0=0$), since
$b_{ke_\alpha}$ is the sum of the $k$ largest entries of column
$\alpha$; so the candidate
values in each coordinate are known, and one searches only their at
most $n^r$ products. If the entries are integers, all breakpoints of the piecewise-linear
function $G_M$ (the points where its slope changes) lie at
integer points and the limits may be replaced by evaluations at
$t=\varepsilon=\tfrac12$; the orbit is then recovered exactly by
finitely many evaluations of $G_M$. For general real entries the same
holds with $t$ and $\varepsilon$ any positive number smaller than the
least gap between two distinct entries lying in a common column.
\end{remark}

The procedure of \cref{rem:algorithm} is summarized as
\cref{alg:reconstruct}, which writes $G$, $N$ for the transforms $G_M$,
$N_M$ and as stated is not optimized for speed: the
candidate set $V_1\times\dots\times V_r$ has up to $n^r$ points (each
column contributes at most $n$ distinct values), and the
inclusion--exclusion step costs $2^r$ evaluations of $N_M(\lambda)$ at
each, so the running time grows exponentially in the number of columns
$r$, even though the multiset has only $n$ points. There is, however,
considerable room for acceleration: when some column $\alpha$ has
pairwise distinct entries --- the generic case, and the situation
already exploited in \cite[Proposition~13]{Kubo} --- each of its $n$
values belongs to a unique row, and differencing $N_M(\lambda)$ in the
$\alpha$-direction between that value and the next-smaller entry of
column $\alpha$ isolates the row: the resulting difference is an
indicator function whose unique breakpoint in each remaining column is
the row's coordinate there. This recovers the orbit in time polynomial in $n$
and $r$ without enumerating the full candidate set; we do not pursue
such refinements here.

\begin{algorithm}[ht]
\caption{Reconstruct the point multiset $[M]$ from the basic values.}
\label{alg:reconstruct}
\begin{algorithmic}[1]
\Require the $\binom{n+r}{r}$ basic values $b_c$ \ ($c\in\N^r$, $|c|\le n$)
\Ensure the multiset $[M]=\{m_1,\dots,m_n\}\subset\R^r$
\State $G(\lambda)\gets\max_{|c|\le n}\bigl(b_c-\langle c,\lambda\rangle\bigr)$
       \Comment{\cref{lem:duality}}
\For{$\alpha=1$ \textbf{to} $r$}
  \State $V_\alpha\gets\{\,b_{ke_\alpha}-b_{(k-1)e_\alpha}:k=1,\dots,n\,\}$
         \Comment{entries of column $\alpha$; $b_0=0$}
\EndFor
\State choose $\varepsilon>0$ below the least gap between two distinct
       entries within a single $V_\alpha$ \Comment{$\varepsilon=\tfrac12$ for integer data}
\State $N(\lambda)\gets n-\dfrac{G(\lambda)-G(\lambda+\varepsilon\mathbf 1)}{\varepsilon}$
       \Comment{\cref{lem:counting}; limit attained at $t=\varepsilon$}
\State $[M]\gets\emptyset$
\ForAll{candidates $v\in V_1\times\dots\times V_r$}
       \Comment{at most $n^r$ points}
  \State $m\gets\sum_{S\subseteq\{1,\dots,r\}}(-1)^{|S|}\,N(v-\varepsilon\mathbf 1_S)$
         \Comment{\cref{lem:cdf}}
  \If{$m>0$}
    \State add $v$ to $[M]$ with multiplicity $m$
  \EndIf
\EndFor
\State \Return $[M]$
\end{algorithmic}
\end{algorithm}

Beyond mere injectivity, the coordinate map is stable. Equip $\R^{nr}$
with the Frobenius norm $\|M\|=(\sum_{i,\alpha}m_{i\alpha}^2)^{1/2}$ and
the quotient $\R^{nr}/S_n$ with the induced optimal-matching metric
\[
d([M],[\bar M])=\min_{\pi\in S_n}\|M-\pi\cdot\bar M\|,
\]
where $\pi$ permutes rows; this is the distance obtained by optimally
matching the two underlying point multisets. Write
$\Phi(M)=(b_c(M))_{c\in\N^r,\,|c|\le n}\in\R^{\binom{n+r}{r}}$, with the Euclidean
norm $\|\cdot\|_2$ on the target.

\begin{proposition}[Stability]\label{prop:stability}
The coordinate map $\Phi$ is Lipschitz: for all $M,\bar M\in\R^{nr}$,
\[
\|\Phi(M)-\Phi(\bar M)\|_2\;\le\;
\Bigl(\sum_{c\in\N^r,\,|c|\le n}|c|\Bigr)^{1/2}\, d([M],[\bar M]).
\]
In particular $\Phi$ descends to a Lipschitz injection of
$\R^{nr}/S_n$. For $r=1$ it is moreover a bi-Lipschitz embedding:
$d([x],[\bar x])\le 2\,\|\Phi(x)-\Phi(\bar x)\|_2$.
\end{proposition}

\begin{proof}
Each $b_c$ is a maximum of linear functionals
$M\mapsto\sum_{\alpha}\sum_{i\in S_\alpha}m_{i\alpha}$ whose gradients
are $0/1$ matrices with exactly $|c|$ ones; hence $b_c$ is convex with
subgradients of Frobenius norm at most $\sqrt{|c|}$, so
$|b_c(M)-b_c(\bar M)|\le\sqrt{|c|}\,\|M-\bar M\|$. As $b_c$ is
$S_n$-invariant, $\bar M$ may be replaced by $\pi\cdot\bar M$ for the
optimal $\pi$, giving $|b_c(M)-b_c(\bar M)|\le\sqrt{|c|}\,d([M],[\bar
M])$. Squaring and summing over $c$ yields the stated bound, with
Lipschitz constant $\bigl(\sum_{c\in\N^r,\,|c|\le n}|c|\bigr)^{1/2}
\le\sqrt{n\tbinom{n+r}{r}}$ since $|c|\le n$.

For $r=1$, write $x^{\downarrow}$ for the decreasing rearrangement of
$x$; then $b_k(x)=\sum_{i=1}^k x^{\downarrow}_i$ is the sum of the $k$
largest entries, so $(b_1,\dots,b_n)=Px^{\downarrow}$, where $P$ is the lower-triangular
all-ones (partial-sum) matrix. The quotient metric $d$ on $\R^n/S_n$
equals $\|x^{\downarrow}-\bar x^{\downarrow}\|_2$, and the difference
operator $P^{-1}=I-S$ ($S$ the down-shift) satisfies
$\|P^{-1}\|_2=\|I-S\|_2\le\|I\|_2+\|S\|_2=2$; hence
$d([x],[\bar x])=\|x^{\downarrow}-\bar x^{\downarrow}\|_2
\le 2\,\|P(x^{\downarrow}-\bar x^{\downarrow})\|_2
=2\,\|\Phi(x)-\Phi(\bar x)\|_2$.
\end{proof}

The lower bound extends to all $r$ once one observes that $\Phi$ is a
\emph{max filter bank} in the sense of \cite{BT,CIMP}. Indeed, for
$c\ne0$ let $Z_c$ be an $n\times r$ matrix of $0$s and $1$s with
$c_\alpha$ ones in column $\alpha$, placed in pairwise disjoint rows
(possible as $|c|\le n$).
As $\sigma$ ranges over $S_n$, the support of $\sigma\cdot Z_c$ runs
through exactly the disjoint tuples $(S_1,\dots,S_r)$ with
$|S_\alpha|=c_\alpha$, so
\begin{equation}\label{eq:maxfilter}
b_c(M)=\max_{\sigma\in S_n}\langle M,\ \sigma\cdot Z_c\rangle,
\end{equation}
a max filter with nonzero window $Z_c$ for the finite group $S_n$
acting on $\R^{nr}$ by orthogonal transformations (the row permutations
act by permutation matrices, which are orthogonal).

\begin{corollary}\label{cor:bilipschitz}
For all $n\ge1$ and $r\ge1$, the coordinate map $\Phi$ is a
bi-Lipschitz embedding of $\R^{nr}/S_n$: there is a constant
$c_{n,r}>0$ with
\[
c_{n,r}\,d([M],[\bar M])\;\le\;\|\Phi(M)-\Phi(\bar M)\|_2\;\le\;
\Bigl(\sum_{c\in\N^r,\,|c|\le n}|c|\Bigr)^{1/2} d([M],[\bar M])
\]
for all $M,\bar M\in\R^{nr}$.
\end{corollary}

\begin{proof}
The upper bound is \cref{prop:stability}. For the lower bound we may
assume $nr\ge2$ (the case $n=r=1$ is immediate). By
\eqref{eq:maxfilter} the family $(b_c)_{c\in\N^r,\,0<|c|\le n}$ is a max filter
bank with nonzero windows for the finite group $S_n$ acting
orthogonally on $\R^{nr}$, and by \cref{thm:main} it is injective on
$\R^{nr}/S_n$; dropping the constant value $b_0\equiv0$ affects neither
the injectivity nor $\|\Phi(M)-\Phi(\bar M)\|_2$. Hence
\cite[Corollary~1.5]{BT} applies and yields a constant $c_{n,r}>0$ as
claimed.
\end{proof}

For $r=1$ the constant of \cref{prop:stability} is explicit and
dimension-free ($c_{n,1}=\tfrac12$); for $r\ge2$ the lower bound of
\cite{BT} rests on a compactness argument and does not exhibit
$c_{n,r}$, so the growth of the optimal distortion remains open
(\cref{sec:discussion}).

\section{How much of the basic family is needed?}\label{sec:pairwise}

\Cref{thm:main} uses the full basic family. In this section we
show that certain proper subfamilies do not suffice in general: the
pairwise-supported values separate exactly when $n\le3$
(\cref{prop:n3}), and even together with the top-degree values
they fail for $n\ge8$ (\cref{ex:switching}); on the other hand,
each of these subfamilies carries a geometric meaning
(projections of the row multiset), which we make precise.

For a multiset $A$ of points of a product set
$X_1\times\dots\times X_r$ --- the orbit space $\R^{nr}/S_n$ being the
special case $X_1=\dots=X_r=\R$ with $|A|=n$ --- and a subset
$S\subseteq\{1,\dots,r\}$, we write $\proj_S(A)$ for the projection of
$A$ to the coordinates in $S$, \emph{counted with multiplicity}. The
values supported on two columns are governed by the case $r=2$ of the
main theorem.

\begin{corollary}\label{cor:pairwise}
Let $r\ge2$ and let $\alpha,\beta\in\{1,\dots,r\}$ be distinct. The basic values $b_c(M)$ with
$\supp(c)\subseteq\{\alpha,\beta\}$ determine the projection
$\proj_{\{\alpha,\beta\}}([M])$.
\end{corollary}

\begin{proof}
The values $b_c$ with $\supp(c)\subseteq\{\alpha,\beta\}$ are exactly
the basic $2$-symmetric values of the $n\times2$ matrix
$(x^{(\alpha)},x^{(\beta)})$. By \cref{thm:main} applied with
$r=2$, these values separate the orbits of $S_n$ on $\R^{2n}$, hence
determine
$[(x^{(\alpha)},x^{(\beta)})]=\proj_{\{\alpha,\beta\}}([M])$.
\end{proof}

\subsection{Pairwise values suffice precisely for $n\le3$}

The positive part rests on two combinatorial lemmas. Passing from the
basic values to the pairwise projections (\cref{cor:pairwise})
uses the real structure through the main theorem, but recovering a
multiset from its pairwise projections, the content of the two lemmas,
uses only equality of coordinates; we therefore state them for
arbitrary sets $X_\alpha$. The case of two points is straightforward.

\begin{lemma}\label{lem:twopoints}
Let $r\ge2$, let $X_1,\dots,X_r$ be sets, and let $A,B$ be multisets
of two points of $X_1\times\dots\times X_r$ whose pairwise projections
coincide. Then $A=B$.
\end{lemma}

\begin{proof}
Write $A=\{a,a'\}$ and $B=\{b,b'\}$. For each coordinate $\alpha$,
projecting onto the single coordinate $\alpha$ (through any pairwise
projection involving $\alpha$, which exists since $r\ge2$) yields
\begin{equation}\label{eq:single}
\{a_\alpha,a'_\alpha\}=\proj_{\{\alpha\}}(A)=\proj_{\{\alpha\}}(B)
=\{b_\alpha,b'_\alpha\}.
\end{equation}

If $a=a'$, then \eqref{eq:single} gives $\{b_\alpha,b'_\alpha\}
=\{a_\alpha,a_\alpha\}$ for every $\alpha$, so $b=b'=a$ and $B=A$.

Suppose $a\ne a'$, and fix a coordinate $\gamma$ with
$a_\gamma\ne a'_\gamma$. By \eqref{eq:single} the elements of
$\{b_\gamma,b'_\gamma\}=\{a_\gamma,a'_\gamma\}$ are distinct, so we may
label the points of $B$ so that $b_\gamma=a_\gamma$ and
$b'_\gamma=a'_\gamma$. For every other coordinate $\alpha$, the equal
multisets $\proj_{\{\gamma,\alpha\}}(A)$ and $\proj_{\{\gamma,\alpha\}}(B)$
each contain exactly one point with $\gamma$-coordinate $a_\gamma$ (as
$a_\gamma\ne a'_\gamma$), coming from $a$ in $A$ and from $b$ in $B$;
equality forces $b_\alpha=a_\alpha$, and the other point of each
multiset --- the one with $\gamma$-coordinate $a'_\gamma$, namely $a'$
in $A$ and $b'$ in $B$ --- gives $b'_\alpha=a'_\alpha$. Hence $b=a$,
$b'=a'$, and $B=A$.
\end{proof}

\begin{lemma}\label{lem:threepoints}
Let $X_1,X_2,X_3$ be sets and let $A,B$ be multisets of three points of
$X_1\times X_2\times X_3$ whose three pairwise projections coincide.
Then $A=B$.
\end{lemma}

\begin{proof}
Write $A=\{(x_i,y_i,z_i):i\in\{1,2,3\}\}$. Since
$\proj_{\{1,2\}}(A)=\proj_{\{1,2\}}(B)$, we may index $B$ as
$B=\{(x_i,y_i,z'_i):i\in\{1,2,3\}\}$. From
$\proj_{\{1,3\}}(A)=\proj_{\{1,3\}}(B)$ there is $\pi\in S_3$ with
\begin{equation}\label{eq:pi}
x_i=x_{\pi(i)},\qquad z'_i=z_{\pi(i)}\qquad(i\in\{1,2,3\}),
\end{equation}
and from $\proj_{\{2,3\}}(A)=\proj_{\{2,3\}}(B)$ there is $\rho\in S_3$ with
\begin{equation}\label{eq:rho}
y_i=y_{\rho(i)},\qquad z'_i=z_{\rho(i)}\qquad(i\in\{1,2,3\}).
\end{equation}
It suffices to find $\tau\in S_3$ with $x_i=x_{\tau(i)}$,
$y_i=y_{\tau(i)}$ and $z'_i=z_{\tau(i)}$ for all $i$; then
$(x_i,y_i,z'_i)=(x_{\tau(i)},y_{\tau(i)},z_{\tau(i)})$ and $B=A$. Write
$\mathrm{id}$ for the identity permutation.

If $z'_i=z_i$ for all $i$, then we may take $\tau=\mathrm{id}$: $x_i=x_i$
and $y_i=y_i$ hold by the indexing of $B$, and $z'_i=z_i$ by assumption.

Otherwise $z'_j\ne z_j$ for some $j$, so $\pi\ne\mathrm{id}$ by
\eqref{eq:pi}, and $\pi$ is either a $3$-cycle or a transposition.

\noindent
(i) \emph{$\pi$: a $3$-cycle.} The equalities $x_i=x_{\pi(i)}$ in
\eqref{eq:pi} force $x_1=x_2=x_3$, and we may take $\tau=\rho$, by
\eqref{eq:rho}.

\noindent
(ii) \emph{$\pi$: a transposition.} After a common relabelling of the indices
we may assume $\pi=(1\,2)$, so $x_1=x_2$ and
$(z'_1,z'_2,z'_3)=(z_2,z_1,z_3)$; as $z'_j\ne z_j$ for some $j$, this
forces $z_1\ne z_2$. We claim
$y_1=y_2$; then we may take $\tau=\pi$. The equalities $z'_i=z_{\rho(i)}$
in \eqref{eq:rho}, combined with $(z'_1,z'_2,z'_3)=(z_2,z_1,z_3)$, read
$z_{\rho(1)}=z_2$, $z_{\rho(2)}=z_1$, $z_{\rho(3)}=z_3$. Running
through the six $\rho\in S_3$: $\rho=\mathrm{id}$, $(1\,3)$ and $(2\,3)$
each force $z_1=z_2$, a contradiction; $\rho=(1\,2)$ gives
$y_1=y_{\rho(1)}=y_2$ directly; and each of the two $3$-cycles forces
$y_1=y_2=y_3$ through the equalities $y_i=y_{\rho(i)}$ in
\eqref{eq:rho}. In every feasible case $y_1=y_2$, as claimed.
\end{proof}

\begin{proposition}\label{prop:n3}
Let $r\ge2$ and $n\le3$. The basic values supported on two columns
separate the orbits of $S_n$ on $\R^{nr}$. For $n\ge4$ they do not,
already for $r=3$.
\end{proposition}

\begin{proof}
Suppose all pairwise-supported values agree for $M$ and $\bar M$; by
\cref{cor:pairwise}, $\proj_{\{\alpha,\beta\}}([M])
=\proj_{\{\alpha,\beta\}}([\bar M])$ for all distinct $\alpha,\beta$. For $n=1$
the single row is its own orbit and is recovered from the values
$b_{e_\alpha}$. For $n=2$, apply \cref{lem:twopoints}. For
$n=3$, we prove by induction on $r\ge2$ that two multisets of three
points of $X_1\times\dots\times X_r$ with equal pairwise projections
are equal; the base case $r=2$ is trivial. For $r\ge3$ and each
$(r-1)$-element subset $S$ of coordinates, the projections
$\proj_S([M])$ and $\proj_S([\bar M])$ have equal pairwise
projections, hence they are equal by induction. Now apply
\cref{lem:threepoints}, grouping the $r$ coordinates into the three
blocks $\{1,\dots,r-2\}$, $\{r-1\}$ and $\{r\}$ (so that
$X_1\times\dots\times X_r=(X_1\times\dots\times X_{r-2})\times
X_{r-1}\times X_r$): its three pairwise hypotheses are the equalities
of the projections to $\{1,\dots,r-1\}$ and to $\{1,\dots,r-2,r\}$
(both by the inductive step) and to $\{r-1,r\}$ (a pairwise
projection).

For the negative part, consider the multisets\footnote{A standard
non-uniqueness example in discrete tomography \cite{GG,HT}.}
\begin{equation}\label{eq:paritybox}
\begin{aligned}
A&=\{(0,0,0),(0,1,1),(1,0,1),(1,1,0)\},\\
B&=\{(0,0,1),(0,1,0),(1,0,0),(1,1,1)\}
\end{aligned}
\end{equation}
(in $A$ the third coordinate is the parity of the first two and in $B$
its negation; equivalently, $A$ and $B$ are the even and odd vertices of
the cube).
All three pairwise projections of $A$ and $B$ coincide, while
$A\ne B$. One checks directly that $b_{2,1,1}(A)=3\ne4=b_{2,1,1}(B)$,
in accordance with \cref{thm:main}; incidentally
$b_{1,1,1}(A)=b_{1,1,1}(B)=3$, so the lowest-degree genuinely
three-column invariant does not detect this pair. Padding $A$ and $B$
with common rows gives examples for every $n\ge4$, and appending zero
columns gives examples for every $r\ge3$.
\end{proof}

\subsection{Top-degree values and the difference matrix}

For $|c|=n$ the sets $S_1,\dots,S_r$ in \eqref{eq:bvalue} partition
$\{1,\dots,n\}$, which yields a second exact interpretation of a
subfamily of the basic values. Let the \emph{difference matrix}
$M^\Delta$ be the $n\times(r-1)$ matrix whose $i$-th row is
$\bigl(m_{i2}-m_{i1},\dots,m_{ir}-m_{i1}\bigr)$.

\begin{proposition}\label{prop:topdegree}
Let $r\ge2$. The top-degree basic values $\{b_c(M):|c|=n\}$ determine,
and are determined by, the pair consisting of the first-column sum
$b_{ne_1}(M)=\sum_{i=1}^n m_{i1}$ and the orbit $[M^\Delta]$.
\end{proposition}

\begin{proof}
For $|c|=n$, writing $c=(c_1,\tilde c)$ with
$\tilde c=(c_2,\dots,c_r)$, every row lies in exactly one $S_\alpha$;
adding and subtracting $m_{i1}$ on each $S_\alpha$ gives
\[
b_c(M)=\max\Bigl\{\sum_{i=1}^n m_{i1}
+\sum_{\alpha=2}^{r}\sum_{i\in S_\alpha}(m_{i\alpha}-m_{i1})\Bigr\}
=b_{ne_1}(M)+b_{\tilde c}(M^\Delta).
\]
Here the first sum is constant, since the $S_\alpha$ partition
$\{1,\dots,n\}$, so it factors out of the maximum; maximizing the
second sum over disjoint $S_2,\dots,S_r$ of sizes $c_2,\dots,c_r$ is
exactly $b_{\tilde c}(M^\Delta)$. As $c$ ranges over $\{c:|c|=n\}$,
$\tilde c$ ranges over all of $\{\tilde c\in\N^{r-1}:|\tilde c|\le n\}$.
In particular $b_{ne_1}(M)$ is the top-degree value with $\tilde c=0$,
so the top-degree values determine $b_{ne_1}(M)$ and all basic values
of $M^\Delta$, hence $[M^\Delta]$ by \cref{thm:main} (applied
with $r-1$ columns). The converse direction is immediate from the displayed
identity.
\end{proof}

The reduction $b_c(M)=b_{ne_1}(M)+b_{\tilde c}(M^\Delta)$ already
appears in the proof of \cite[Proposition~14]{Kubo}, where it yielded
separation only on the generic set of matrices for which the
difference of some two columns has pairwise distinct entries;
the present proposition removes that restriction. The choice of column $1$ as
reference is immaterial: the analogous statement holds with any column
taken as the reference.

Geometrically, $[M^\Delta]$ is the projection of the row multiset
along the direction $(1,\dots,1)$. Together with
\cref{cor:pairwise}, the two subfamilies record complementary
projections of the row multiset: the pairwise values give its
projections to the coordinate planes, while the
top-degree values give its projection along $(1,\dots,1)$. For $r=3$
each coordinate-plane projection forgets a single coordinate, so the
three of them are the X-rays along $e_1,e_2,e_3$, which with the
diagonal give the four directions of a discrete tomography problem, and
even together they still do not determine the orbit:

\begin{example}\label{ex:switching}
Let $n=8$, $r=3$, and take the two $8$-point
multisets\footnote{As with \eqref{eq:paritybox}, $A$ and $B$ are such a
non-uniqueness example, here for the four directions
$e_1,e_2,e_3,(1,1,1)$.}
\begin{align*}
A&=\{(0,0,0),(1,1,0),(1,0,1),(0,1,1),(2,1,1),(1,2,1),(1,1,2),(2,2,2)\},\\
B&=\{(1,0,0),(0,1,0),(0,0,1),(1,1,1),(1,1,1),(2,2,1),(2,1,2),(1,2,2)\}.
\end{align*}
One checks directly that $A$ and $B$ have the same three pairwise
projections $\proj_{\{\alpha,\beta\}}$ and the same diagonal projection
$[M^\Delta]$. Hence by
\cref{cor:pairwise} and \cref{prop:topdegree},
\emph{all pairwise-supported and all top-degree basic values of $A$ and
$B$ coincide}. Nevertheless $b_{1,1,2}(A)=8\ne7=b_{1,1,2}(B)$. By definition
$b_{1,1,2}(M)$ selects four distinct rows and adds the first entry of
one, the second entry of another, and the third entries of the
remaining two, maximizing the total. In $A$ the four distinct rows
$(2,1,1),(1,2,1),(1,1,2),(2,2,2)$ supply the value $2$ in all four
slots, giving $2+2+2+2=8$. In $B$ the only rows with third entry $2$
are $(2,1,2)$ and $(1,2,2)$, so both are forced into the two
third-entry slots; among the remaining rows only $(2,2,1)$ still
carries a $2$ (in its first and second entries), and it can fill just
one of the two remaining slots, so the other yields at most $1$. The
maximum is therefore $2+2+2+1=7$. Padding $A$ and $B$ with common rows
gives such configurations for every $n\ge8$.
Thus any separating family must include a value supported on at least
three columns and of degree less than $n$: neither the pairwise values
(those $b_c$ with $\supp(c)\subseteq\{\alpha,\beta\}$ for some pair
$\{\alpha,\beta\}$, of any degree) nor the top-degree values
($|c|=n$), nor their union, separates orbits in general.
\end{example}

\section{Concluding remarks}\label{sec:discussion}

\paragraph{Size and degree of the separating family}
\Cref{thm:main} exhibits a separating family of
$\binom{n+r}{r}$ invariants of degree at most $n$ for the row
permutation action of $S_n$ on $\R^{nr}$. The general construction of
\cite{Derksen} for this action yields $nr+(nr)!/n!$ separating
invariants of degree at most $\max(nr,\binom{nr}2)$. The basic family is thus
considerably smaller and of much lower degree, but it is not minimal as
it stands --- $b_0$ is constant, and by \cref{prop:n3} the values
supported on at least three columns are redundant when $n\le3$ --- and several
questions about economy remain open. First, \emph{which} basic values
are necessary: by \cref{prop:n3} and
\cref{ex:switching} any separating subfamily must contain
values supported on at least three columns and of degree less than $n$, yet a
precise description of the minimal separating subfamilies
is unknown. Second, \emph{how few} invariants of any kind suffice:
since $\R^{nr}/S_n$ is $nr$-dimensional, any continuous separating map
needs at least $nr$ coordinates, leaving a wide gap to the
$\binom{n+r}{r}$ basic values, and the minimal cardinality of a
separating family is not known. Third, beyond such counts, whether the basic family admits a
structural characterization --- for instance, minimality among families of
invariants that are defined uniformly in $n$ and symmetric in the $r$
columns --- remains to be formulated. Since the main theorem shows the
basic family to be a complete set, generalizing the tropical
elementary symmetric polynomials to which it reduces when $r=1$, it
would now be reasonable, following the suggestion of \cite{Kubo}, to
call these polynomials \emph{elementary}.

\paragraph{Rational invariants}
The second question raised in the Discussion of \cite{Kubo} --- whether
the basic polynomials generate the semifield of $r$-symmetric tropical
rational functions --- remains open. Derksen
\cite[Theorem~1.2]{Derksen} shows that the invariant semifield is
generated by invariant polynomials of degree at most
$N p_1p_2\cdots p_k$ ($N=nr$, $k=n!$, $p_i$ the $i$-th prime), so a
finite generating set exists; whether the basic family suffices is an
attractive test case.

\paragraph{Bi-Lipschitz embeddings}
Since $\Phi$ is an injective max filter bank \eqref{eq:maxfilter}, it
is a bi-Lipschitz embedding for all $n$ and $r$
(\cref{cor:bilipschitz}), by the bi-Lipschitz theory of max filtering
\cite{BT,CIMP}. What remains open is \emph{quantitative}: the lower
bound of \cite{BT} for $r\ge2$ comes from a compactness argument and
does not exhibit the constant, whereas for $r=1$ we have the explicit
dimension-free value $c_{n,1}=\tfrac12$ (\cref{prop:stability}). Is
there an explicit $c_{n,r}>0$ with $c_{n,r}\,d([M],[\bar M])\le
\|\Phi(M)-\Phi(\bar M)\|_2$ for all $M,\bar M$, and how does the
optimal distortion grow with $n$ and $r$? The reconstruction of
\cref{sec:main} passes through the nonlinear steps
$\Phi\mapsto G_M\mapsto N_M\mapsto[M]$, so a quantitative inverse bound
does not follow from injectivity alone, which is only qualitative;
obtaining it requires controlling the stability of these nonlinear
steps. We leave this --- and its relevance to
permutation-invariant coordinates for point multisets \cite{CIMP,CK}
--- for future work.

\paragraph{Empirical distribution and order statistics}
For $r=1$, \cref{rem:r1} recovers the counting function $N_M$
($n$ times the empirical distribution function of $x_1,\dots,x_n$) from
the basic values, which are thus equivalent to the order statistics of
the data. For general $r$, the same reconstruction returns
$N_M(\lambda)=\#\{i:\ m_{i\alpha}\le\lambda_\alpha\ (\forall\alpha)\}$,
that is, $n$ times the multivariate (joint) empirical distribution
function of the $r$-attribute data, obtained from the transform
$G_M(\lambda)=\sum_i\max\bigl(0,\max_\alpha(m_{i\alpha}-\lambda_\alpha)\bigr)$.
Thus the basic values $b_c(M)$ form a \emph{complete invariant} of the
multiset: they determine its joint empirical distribution, and hence
the data itself. This admits a statistical reading: if the rows are
regarded as a sample from an unknown distribution $Q$ on $\R^r$, the
likelihood depends on the data only through the multiset of rows, so
any quantity that determines the multiset --- in particular the basic
values $b_c(M)$ --- is a \emph{sufficient statistic} for $Q$, in the
spirit of the tropical sufficient statistics of \cite{MKPC}. This is
the multivariate counterpart of the order-statistics interpretation in
the single-attribute case \cite{Kubo,Matsui}: while $\R^r$ admits no
canonical notion of order statistics, the joint empirical distribution
function is canonical, and the basic values recover it as the counting function $N_M$, obtained from the transform $G_M$ (\cref{lem:counting}). A
finer study of the basic values as stable, complete features for
multi-attribute data we leave for future work.

\bigskip
\begin{center}
\emph{Dedicated to the memory of Kiyoshi Matsui.}
\end{center}


\subsection*{Declarations}

During the preparation of this work the author used Claude (Anthropic)
to refine the English exposition and to check the clarity of the
presentation. After using this tool, the author reviewed and edited the
content as needed. The author assumes responsibility for all content.

\bibliographystyle{siamplain}
\bibliography{references}

\end{document}